\documentclass[twocolumn,prl, superscriptaddress,longbibliography,margin=1mm 11pt]{revtex4-2} 
\usepackage{mathtools}
\usepackage{amssymb}
\usepackage{amsfonts}
\usepackage{fancyhdr}
\usepackage{slashed}
\usepackage{bbm}
\usepackage{latexsym,epsfig,bbm}
\usepackage[colorlinks=true,urlcolor=blue,citecolor=blue]{hyperref}
\usepackage{amsthm}

\def \d {\mathrm{d}}

\usepackage{color}
\definecolor{blue}{rgb}{0,0.2,1}

\definecolor{red}{rgb}{0.9,0,0}

\newcommand{\bb}{\boldsymbol}

\newtheorem{theorem}{Theorem}
\newtheorem{lemma}[theorem]{Lemma}
\newtheorem{definition}[theorem]{Definition}

\def \d {\mathrm{d}}
\def \e {\mathrm{e}}
\def \i {\mathrm{i}}

\begin{document}

\title{Quantum simulation of partial differential equations via Schr\"odingerisation}

\author{Shi Jin}
\affiliation{Institute of Natural Sciences, School of Mathematical Sciences, MOE-LSC, Shanghai Jiao Tong University, Shanghai, 200240, P. R. China}
\affiliation{ Shanghai Artificial Intelligence Laboratory, Shanghai, China}

\author{Nana Liu}
\email{nana.liu@quantumlah.org}
\affiliation{Institute of Natural Sciences, School of Mathematical Sciences, MOE-LSC, Shanghai Jiao Tong University, Shanghai, 200240, P. R. China}
\affiliation{ Shanghai Artificial Intelligence Laboratory, Shanghai, China}
\affiliation{University of Michigan-Shanghai Jiao Tong University Joint Institute, Shanghai 200240, China}
\author{Yue Yu}
\affiliation{Institute of Natural Sciences, School of Mathematical Sciences, MOE-LSC, Shanghai Jiao Tong University, Shanghai, 200240, P. R. China}

\date{\today}

\begin{abstract} 
We present a simple new way--called {\it Schr\"odingerisation}-- to simulate general linear partial differential equations via quantum simulation. Using a simple new transform, referred to as {\it the warped phase transformation}, any linear partial differential equation can be recast into a system of Schr\"odinger's equations – in real time — in a straightforward way. This can be seen directly on the level of the dynamical equations without more sophisticated methods. This approach is not only applicable to PDEs for classical problems but also those for quantum problems -- like the preparation of quantum ground states, Gibbs states and the simulation of quantum states in random media in the semiclassical limit. 
\end{abstract}
\maketitle 

\section{Introduction}

Quantum simulation is one of the most natural tasks for a quantum device — by preparing outputs of Schr\"odinger's equations directly via its own evolution \cite{feynman2018simulating}. However, what can we say about the suitability of quantum systems in simulating classical dynamical laws, typically in the form of ordinary or partial differential equations (ODEs/PDEs) that are not Schr\"odinger's equations? This is of great importance for broader applications of quantum computing in science and engineering. 

While the most obvious applications are for classical problems, simulating classical dynamics is also important for \textit{quantum problems}. The most well-known include the preparation of quantum ground states \cite{ge2019faster, lin2020near} and Gibbs states \cite{poulin2009sampling}, which benefit from non-unitary dynamics, and are particularly important in areas like quantum chemistry \cite{lloyd1996universal} and optimisation. 

We know quantum simulation of classical dynamics is in principle possible, if we accept reductionism and that the world is fundamentally quantum mechanical. All classical dynamical laws are, with the exception of relativity, in principle derivable from underlying Schr\"odinger's equations. From a computational perspective, we also know that any classical gate can be trivially embedded into a quantum gate, so we expect quantum gates to be able to simulate classical ones. However, the quantum degrees of freedom involved could typically be significantly larger than that of the classical systems one wishes to simulate. An essential question is then:  are there less resource-intensive ways to represent classical dynamics with Schr\"odinger's equations? 

Consider the case of numerical solutions   to linear PDEs, which, upon spatial and temporal discretizations, become a system of linear algebraic equations. Quantum algorithm speedups in solving a system of linear algebraic equations \cite{harrow2009quantum, childs2017quantum} can subsequently lead to possible polynomial or super-polynomial speedups in solving PDEs \cite{childs2021high, montanaro2016quantum}. For nonlinear counterparts see \cite{liu2021efficient,jin2022quantum}. However, here quantum simulation only plays the role of an algorithmic primitive — due to discretisation in time, one does not prepare solution states continuously in time $t$ by evolution of a Schr\"odinger equation in time $t$.  

 To obviate the discretisation in time, a main difficulty is in finding a way to represent non-unitary dynamics with a unitary one, which can be evolved  by quantum simulation like the Schr\"odinger equation. One way is via qubitisation \cite{low2019hamiltonian, gilyen2019quantum} (or block-encoding) and this has very recently been applied to linear PDEs \cite{an2022theory}. This involves unitary dilation methods (e.g. \cite{horn2012matrix}) and has origins in quantum signal processing \cite{low2016methodology}. While this fairly general formalism can in principle approximate the action of any non-unitary operator, it relies heavily on building linear combinations of quantum states and finding suitable polynomial approximations to the non-unitary operator, which is not always simple to describe or to implement in practise. 

For the heat equation, another method is the imaginary time evolution approach, where the heat equation can be converted to a Schr\"odinger equation through a change to \text \it{imaginary} time $t \rightarrow it$ \cite{lehtovaara2007solution}. However, the state obeying Schr\"odinger's equation and its imaginary time counterpart do not have the same evolution. This means that, except for  the steady state solution, extra resources are necessary to map between the solution in the unitarily evolving system to the other. For instance, one requires tomographic measurements of quantum states at each small time step in the unitary evolution \cite{motta2020determining}. Furthermore, the imaginary time evolution method does not apply beyond the heat equation.

Thus we have an important question: \textit{Is there a simpler and generic way of obtaining Schr\"odinger's equations naturally from any given linear dynamics?}

We propose a new paradigm — based on a simple transformation called the \textit{warped phase transformation} — that can map any linear PDE to Schr\"odinger equations in real time. By the quantum simulation of this Schr\"odinger dynamics in time $t$, it is possible to prepare solutions of the original PDEs at any time $t$, in the form of a quantum state. We call this the \textit{Schr\"odingerisation approach}.

To illustrate this approach, we present the heat equation example in detail. With a warped phase transformation and a Fourier transform, the heat equation can be transformed into a system of uncoupled Schr\"odinger equations with rescaled time coordinates — \textit{in real time} — thus providing a conceptual alternative to imaginary time evolution methods. We show how the quantum simulation of these Schr\"odinger's equations can be used to prepare quantum ground states and Gibbs states.

In another example, we apply the Schr\"odingerisation approach to a linear radiative transport equation. This can be used to simulate quantum systems in random media in the semi-classical limit. 

Finally, we show how to generalise the Schr\"odingerisation approach to any linear PDE.

\section{Background}
Consider linear $(d+1)$-dimensional PDEs for $u=u(t,x)$ that is first-order in time $t \geq 0$
\begin{align}
\partial_t u=\mathcal{L}u, \quad u(0,x)=u_0(x).
\end{align}
Here $x \in \mathbb{R}^d$ is the position in $d$ dimensions, $\partial_{(\cdot)}$ is the partial derivative with respect to $(\cdot)$ and $\mathcal{L}$ is a linear differential operator with respect to $x$. 

Schr\"odinger's equation is such a linear PDE equation where $u$ is called the wavefunction and $\mathcal{L}u=-i(-\nabla_x^2+V(x))u$, where $\nabla_x^2$ is the Laplace operator with respect to $x$ and $V(x)$ is the potential function. To find numerical solutions to this problem, one can  discretise in $x$, with uniform mesh sizes along each dimension $\Delta x=2/M$ where $M$ is a positive even integer. Let $\bb{u}(t)$ represent the $M^d$-dimensional vector whose entries are $u(t,x_i)$ at grid points with $i=1,...,M^d$. By definition of the wavefunction, the $l_2$ norm $\|\bb{u}(t)\|=1$. The entries to this vector can be considered as amplitudes of the $M^d$-dimensional quantum state $|u(t)\rangle=\sum_{i=1}^{M^d}u(t, x_i)|i\rangle$ where $\{|i\rangle\}$ is an orthonormal basis set. When given a state spanned by $\{|0\rangle, |1\rangle \}$ over the field $\mathbb{C}$, this is called a qubit. Thus $|u(t)\rangle$ can be described by a system of $d\log_2(M)$ qubits. It is straightforward to see that $\bb{u}(t)$ satisfies the linear system of ODEs 
\begin{align} \label{eq:schrode1}
i\frac{d\bb{u}(t)}{dt}=\bb{H}\bb{u}(t)
\end{align}
where $\bb{H}$ is a $M^d \times M^d$ Hermitian matrix that results from a discretisation of the Schr\"odinger Hamiltonian $-\nabla^2_x+V(x)$ in $x$. The solution to Eq.~\eqref{eq:schrode1} is clearly $\bb{u}(t)=\exp(-i\bb{H}t)\bb{u}(0)$ and one has the corresponding quantum state evolution $|u(t)\rangle=\exp(-i\bb{H}t)|u(0)\rangle$. Quantum simulation addresses the question of how to realise this $\exp(-i\bb{H}t)$ operator, when given $\bb{H}$, and an estimation of the resources required. We focus on the time-independent setting for now. 

Quantum simulation falls roughly into two categories: analogue and digital \cite{daley2022practical}. In the analogue case, a quantum system that naturally realises $\bb{H}$ or can be mapped to this Hamiltonian is found. However, if such an analogue quantum system cannot be found, digital quantum simulation methods still allows the simulation for more general $\bb{H}$. In this case, one counts resources as the number of queries to some given black-boxes and the number of two-qubit gates needed. This is referred to as the query and gate complexities respectively. 

Let $s$ be the sparsity of $\bb{H}$ (maximum number of non-zero entries in each row) and  $\|\bb{H}\|_{\text{max}}$ be its max-norm (value of largest entry in absolute value). We denote the $(i,j)^{\text{th}}$ entry to $\bb{H}$ as $H_{ij}$. A common set of black-boxes used in Hamiltonian simulation is known as the sparse access.
\begin{definition}
Sparse access to Hermitian matrix $\bb{H}$ refers to two unitary black-boxes $O_M$ and $O_F$ such that $O_M|j\rangle|k\rangle|z\rangle=|j\rangle|k\rangle|z\oplus H_{jk}\rangle$ and $O_{F} |j\rangle|l\rangle=|j\rangle|F(j,l)\rangle$. 
Here the function $F$ takes the row index $j$ and a number $l=1,2,...,s$ and outputs the column index of the $l^{\text{th}}$ non-zero elements in row $j$. 
\end{definition}
There are quantum simulation protocols in terms of query complexity that scale linearly in $t$ \cite{low2019hamiltonian} using sparse access or linearly in $t$ up to logarithmic factors \cite{berry2015hamiltonian}.
\begin{lemma} \cite{berry2015hamiltonian} ~\label{eq:lemmasimulation}
Let $\tau=s t\|\bb{H}\|_{\text{max}}$. Then $\exp(-i\bb{H}t)$ acting on $m_H$ qubits can be simulated to within error $\varepsilon$ has query complexity $\mathcal{O}(\tau \log (\tau/\varepsilon)/(\log\log (\tau/\varepsilon)))$ with gate complexity 
$\mathcal{O}\Big(  \tau ( m_H + \log^{2.5}(\tau/\varepsilon) )\log (\tau/\varepsilon)/(\log\log (\tau/\varepsilon)) \Big)$.
 \end{lemma}
Throughout the paper, we use $\tilde{\mathcal{O}}$ to denote $\mathcal{O}$ where logarithmic terms are ignored.
\section{Schr\"odingerisation of the heat equation} 

We begin with the initial value problem of the linear heat equation with a source term
\begin{equation} \label{eq:heatequation}
\partial_t u=-H u, \\ \quad 
u(0,x) = u_0(x),
\end{equation}
where $H=-\nabla_x^2+V(x)$ and $t\ge 0$. We introduce a real one-dimensional variable $p>0$ and define 
\begin{align}
w(t,x,p) = e^{-p} u(t,x).
\end{align} 
This transformation -- which we call the \textit{warped phase transformation} -- is a crucial ingredient that will allow us to transform the heat equation to Schr\"odinger equations by extending the heat equation solution to a higher `warped' extra dimension.  Such a transformation was used in \cite{golse2022quantum} for a completely different purpose: to develop  efficient quantum computing algorithms for uncertainty quantification problems in PDEs. 

First observe that one can recover the solution to the heat equation  using $u(t,x)= \int_0^\infty w(t,x,p)\d p = \int_{-\infty}^\infty \chi(p) w(t,x,p) \d p$, where $\chi(p)$ is the indicator function with $\chi(p) = 1$ for $p \geq 0$ and $\chi(p) = 0$ for $p<0$. One can also recover $u$ via
$u(t,x)=e^p W(t,x,p)$ for any $p>0$. One can extend the domain of $p$ to $ (-\infty, \infty)$, with  $w=w(t,x,p)$ satisfying $\partial_t w +  (\nabla^2_x-V(x))\partial_p w= 0$ with evenly extended initial condition $w(0,x,p)=\exp(-|p|u_0(x))$. Let $\tilde{w}=\tilde w(t,x,\eta)$ be the Fourier transform of $w$ in $p$ and $\eta\in \mathbb{R}$ be the Fourier mode. Then $\tilde{w}$ satisfies a system of \textit{uncoupled} Schr\"odinger equations 
\begin{align} \label{eq:heatschrodinger}
i \partial_t\tilde{w} = \eta (-\nabla_x^2+V(x))\tilde{w},
\end{align}
one for each $\eta$! See Appendix A for details. We call this the \textit{Schr\"odingerised heat equation}. The role of $\eta$ can be interpreted as a recaling in time with $t \rightarrow t \eta$ for the Schr\"odinger equation $i\partial_t \tilde{w}=(-\nabla^2_x+V(x))\tilde{w}$. The rescaled time remains real-valued. 

To solve these equations numerically, we discretise the system in $x$ and $p$, but not in $t$. We choose uniform mesh sizes $\Delta x = 2/M$ for the position variable in each dimension, $\Delta p = 2L/N$ for the $p$ variable, $\Delta \eta=2L/N$ for the $\eta$ variable, where $M$ and $N$ are even positive integers and $L>0$. We introduce the vector $\bb{w}(t)=\bb{w}_{[x]}\otimes \bb{w}_{[p]}$ where the elements of the vector $\bb{w}_{[x]}$ are $w(t, x_i, p)$ with $i=1,...,M^d$ labelling the $x$ grid points and the elements of  $\bb{w}_{[p]}$ are $w(t,x,p_j)$ with $p_j=j \Delta p$, $j=-N/2,...,N/2$. We define $\tilde{\bb{w}}=(\mathbf{1}^{\otimes M^d}\otimes \mathcal{F}_p)\bb{w}$ where $\mathcal{F}_p$ is the discrete Fourier transform with respect to variable $p$. Then the discretisation of the Schr\"odingerised heat equation in Eq.~\eqref{eq:heatschrodinger} becomes the following system of $M^dN$ ODEs  
\begin{align} \label{eq:schrodinger}
i\frac{d}{dt} \tilde{\bb{w}}(t)=(\bb{H}\otimes \bb{D})\tilde{\bb{w}}(t)=\bb{H}_{\text{total}}\tilde{\bb{w}}(t).
\end{align}
where $\bb{H}_{\text{total}}$ is a Hermitian matrix. Here $\bb{H}=(-\bb{P}^2_1-...-\bb{P}^2_d+\bb{V})$ is also Hermitian, originating from the discretisation of the Schr\"odinger Hamiltonian $H$ in Eq.~\eqref{eq:heatequation}. We define the matrices $\bb{P}_l=\mathbf{1}^{\otimes l-1}\otimes P_{l}\otimes \mathbf{1}^{\otimes d-l}$, $l=1,...,d$,  where $P_{l}$ is the discretisation of the momentum operator $-i\partial_{x}$ with respect to the $l^{\text{th}}$ spatial variable. The matrix $\bb{V}=\text{diag}(V(x_1),...,V(x_{M^d}))$ is a diagonal matrix with entries $V(x_i)$. The matrix $\bb{D}=\text{diag}(\mu_1,...,\mu_N)$ is also a diagonal matrix, with entries $\mu_j=\pi(j-N/2)$.  

It is clear that  $\tilde{\bb{w}}(t)=\exp(-i\bb{H}_{\text{total}}t)\tilde{\bb{w}}(0)$ and $\bb{H}_{\text{total}}$, being Hermitian, can be interpreted as a time-independent Hamiltonian. We can define the quantum state $|\tilde{w}(t)\rangle=(1/\|\tilde{\bb{w}}(t)\|)\sum_{i=1}^{M^d} \sum_{j=-N/2}^{N/2} \tilde{w}(t, x_i, \eta_j)|i,j\rangle$ where $\|\tilde{\bb{w}}(t)\|^2=\sum_{i=1}^{M^d}\sum_{j=-N/2}^{N/2}|\tilde{w}(t, x_i, \eta_j)|^2$ is the $l_2$ norm squared of $\tilde{\bb{w}}(t)$. 
From Eq.~\eqref{eq:schrodinger} one can recover the quantum state  with input quantum state $|\tilde{w}(0)\rangle$ through $|\tilde{w}(t)\rangle=\exp(-i \bb{H}_{\text{total}}t) |\tilde{w}(0)\rangle$. By applying an inverse quantum Fourier transform $\mathcal{F}_p^{-1}$, with respect to $p$, onto the second register we obtain $|w(t)\rangle=(\mathbf{1}^{\otimes M^d} \otimes \mathcal{F}^{-1}_p)|\tilde{w}(t)\rangle$, where $|w(t)\rangle=(1/\|\bb{w}(t)\|)\sum_{i=1}^{M^d}\sum_{j=1}^N w(t, x_i, p_j)|i, j\rangle$ and $\|\bb{w}(t)\|=\|\tilde{\bb{w}}(t)\|$. From $|w(t)\rangle$ one can recover the quantum state of $u$  whose entries are proportional to the solutions of the heat equation $|u(t)\rangle=(1/\|\bb{u}(t)\|)\sum_{i=1}^{M^d}u(t, x_i)|i\rangle$. We can do this  by  either projecting $|w(t)\rangle$ onto $\mathbf{1} \otimes \sum_{k=N/2}^N |k\rangle \langle k|$ (projecting only onto $p>0$), or using amplitude amplification to boost the chance of retrieving $|u(t)\rangle$ to probability $\sim \|\bb{u}(0)\|/\|\bb{u}(t)\|$.

 We call this the \textit{Schr\"odingerisation approach} to prepare $|u(t)\rangle \propto \exp(-\bb{H}t)|u(0)\rangle$, where $|u(t)\rangle$ evolves with respect to $\bb{H}$ instead of $\bb{H}_{\text{total}}$. The above method is in fact applicable to {\it any} Hermitian $\bb{H}$, with sparsity $s$ and max-norm $\|\bb{H}\|_{\text{max}}$. We extend to non-Hermitian $\bb{H}$ in the last section.

\begin{theorem} \label{thm:heatequation} Given sparse-access to a $D \times D$ Hermitian matrix $\bb{H}$ and the unitary $U_{initial}$ that prepares the initial quantum state $|u(0)\rangle=U_{initial}|0\rangle$. With the Schr\"odingerisation approach, the state $|u(t)\rangle$ can be prepared to precision $\epsilon$ with query and gate complexity $\tilde{\mathcal{O}}((\|\bb{u}(0)\|/\|\bb{u}(t)\|)st\|\bb{H}\|_{\text{max}}/\epsilon)$.
\end{theorem}
\begin{proof}
See Appendix B.
\end{proof}

\subsection{Preparing quantum ground states and Gibbs states}
Suppose one wants to prepare a $D$-dimensional ground state $|E_0\rangle$ from a given state $|u(0)\rangle=\sum_{j=0}^{D-1} \alpha_j|E_j\rangle$, where $\alpha_j \in \mathbb{C}$, $\|\bb{u}(0)\|=1$ and $\{|E_j\rangle\}$ are the non-degenerate orthonormal eigenstates of a (Hermitian) $D \times D$ Hamiltonian $\bb{H}$. $\bb{u}(t)$ is the vector whose entries are the amplitudes of $|u(t)\rangle$. If one evolves $\bb{u}(t)=\exp(-\bb{H}t)\bb{u}(0)$ according to $\partial_t \bb{u}=-\bb{H}\bb{u}$ with initial condition $\bb{u}(0)$, then one can write $|u(t)\rangle=(1/\|\bb{u}(t)\|)\sum_{j=0}^{D-1} \alpha_j \exp(-E_j t)|E_j\rangle$. Then $|u(t)\rangle=\exp(-\bb{H}t)|u(0)\rangle$ where $\bb{H}$ is a  $D \times D$ Hermitian matrix with sparsity $s$ and max-norm $\|\bb{H}\|_{\text{max}}$. Assuming the presence of a non-zero spectral gap $\Delta=E_1-E_0>0$, the convergence to the ground state $|u(t)\rangle \rightarrow |E_0\rangle$ is exponentially fast. If $t_{\text{final}}$ denotes the time where the quantum fidelity between the target ground state $|E_0\rangle$ and the prepared state $|u(t)\rangle$ be $F(|E_0\rangle, |u(t)\rangle)=1-\epsilon$ for a small $\epsilon>0$, then the dominating contribution to $t_{\text{final}}$ is $\mathcal{O}((1/\Delta)\ln(1/(\epsilon |\alpha_0|^2)))$. Combining this with Theorem~\ref{thm:heatequation}, one sees the total query and gate complexity costs in preparing the ground state to fidelity $1-\epsilon$ for $\epsilon \ll 1$ is $\tilde{\mathcal{O}}(s\|\bb{H} \|_{\text{max}}/(|\alpha_0|\Delta \epsilon))$. Here the scaling in $|\alpha_0|$ and $\Delta$ is comparable to non-heuristic schemes like quantum phase estimation \cite{ge2019faster} and near-optimal lower bounds for ground state preparation \cite{lin2020near}, up to logarithmic factors. The drawback is that this scheme has an extra factor $1/\epsilon$ compared to the near-optimal schemes. This originates from $\|\bb{H}\otimes \bb{D}\|_{\text{max}}\sim \|\bb{H}\|_{\text{max}}/\epsilon$. It will be interesting to see if improvements to this factor can be achieved. For details see Appendix C. 

To create the Gibbs state at temperature $T$ corresponding to the same $D \times D$ Hamiltonian $\bb{H}$, we can use the Schr\"odingerisation approach to prepare the normalised pure state $|\Psi(\beta)\rangle=\sum_k \sqrt{\exp(-\beta E_k)/Z}|E_k E_k\rangle$, where $Z=\text{Tr}(\exp(-\beta \bb{H}))$ is the partition function for $\bb{H}$ with $\beta=1/(k_BT)$ and $T$ is the temperature. Then by tracing out one register we obtain the Gibbs state $\rho_{Gibbs}(\beta)=\sum_{j=0}^{D-1} \exp(-\beta E_j)|E_j\rangle \langle E_j|=\text{Tr}_1(|\Psi (\beta)\rangle \langle \Psi (\beta)|)$. This quantum simulation method differs from previous methods \cite{poulin2009sampling} that need to make use of quantum phase estimation to prepare $|\Psi(\beta)\rangle$. We rewrite $|\Psi (\beta)\rangle \propto \exp(-(\beta/2)(\bb{H}\otimes \mathbf{1})\sum_{j=0}^{D-1}|E_j E_j\rangle$ where we can identify $|u(t)\rangle=|\Psi(\beta)\rangle$ by the application of $\exp(-(\bb{H} \otimes \mathbf{1}) t)$, with time $t=\beta/2$, onto the state $|u(0)\rangle=(1/\sqrt{D})\sum_{j=0}^{D-1} |E_j E_j\rangle$. The latter we assume is given. This is equivalent to performing the simple extension $\bb{H} \rightarrow \bb{H}\otimes \mathbf{1}$ in Eq.~\eqref{eq:schrodinger}, for any Hermitian $\bb{H}$. Now $\|u(t)\|=\sqrt{Z}$ and $\|u(0)\|=1/\sqrt{D}$ where $D$ is the total number of energy levels. Then a straightforward application of Theorem~\ref{thm:heatequation} means we can prepare $\rho_{Gibbs}(\beta)$ to precision $\epsilon$ with query and gate complexity $\tilde{\mathcal{O}}(s\|\bb{H}\|_{\text{max}} \beta \sqrt{D/Z}/\epsilon)$. Heuristic methods aside, this coincides with the best-known scaling, to our knowledge, with respect to $D$ and $Z$ \cite{poulin2009sampling, gilyen2019quantum}. 

\section{Schr\"odingerisation of linear transport equation}

An interesting application of solving the linear transport equation (non-unitary dynamics) for a quantum problem is in simulating quantum states in random media in the semi-classical regime. Here we are interested in the evolution of a quantum system -- described by the Schr\"odinger equation -- in the presence of a random potential $U(x)=U_0(x)+U_1(x/\zeta)$. Here $U_0(x)$ is a deterministic, slowly moving background potential and $U_1$ is the random fluctuation depending on $x/\zeta$, the fast spatial variable, and $R(|x-y|)=\langle U_1(x)U_2(y)\rangle$ is the covariance of the fluctuations. We assume homogeneity and isotropy of the medium, where the differential scattering cross-section $\sigma(k,k') \propto R(k-k')\delta(k^2-(k')^2)$ with $\tilde{R}(|k|)$ being the power spectrum and $\Sigma(k)$ denotes the total scattering cross-section. In the semiclassical limit, the Wigner transform of the wave function will converge to the particle probability density function  $W=W(t, x, k)$
governed by  the linear transport equation \cite{ryzhik1996transport}
\begin{align} \label{eq:lineartranswigner}
\partial_t W+k \cdot \nabla_x W=\int dk'\sigma(k,k')W(t,x,k')-\Sigma(k) W.
\end{align}
By Schr\"odingerising this equation, we can prepare, to precision $\epsilon$, the quantum state $|W(t)\rangle$ whose amplitudes are proportional $\bb{W}(t)$, which is the vector formed by taking $W$ at mesh points $(x_i, k_j)$ formed by discretising $x$ and $k$. Following the Schr\"odingerisation approach (see Appendix D for details), the corresponding $\bb{H}_{\text{total}}$ for this problem is $\bb{H}_{\text{total}}=\bb{L}_{[\xi, k]} \otimes \bb{1}-\mathbf{1}\otimes \bb{\Sigma} \otimes \bb{D}+\mathbf{1}\otimes \bb{\sigma} \otimes \bb{D}$, where $\bb{\sigma}$ and $\bb{\Sigma}$ are the differential and total cross-section matrices. Matrix $\bb{L}_{[\xi, k]}$ has entries corresponding to $\xi \cdot k$, where $\xi$ is the Fourier mode of $x$ and arises from discretising  $k \cdot \nabla_x W$, which is a fundamental part of the transfer equation in Eq.~\eqref{eq:lineartranswigner}. This part exists whether or not the Schr\"odingerisation approach is used. Here $\|\bb{L}_{[\xi, k]}\|_{\text{max}} \sim \|\bb{\sigma}\otimes \bb{D}\|_{\text{max}} \sim \mathcal{O}(1/\epsilon)$ where $\bb{\sigma}$ and $\bb{\Sigma}$ are order one factors. Since the warped transformation only adds one derivative in $p$ in the scattering term, and the new transformed equation remains first order, then the order of  $\|\bb{H}_{\text{total}}\|_{\text{max}}$ {\it remains the same order in $1/\epsilon$} as the original transport equation. This means that the Schr\"odingerisation approach does not give rise to an extra $O(1/\epsilon)$ factor compared to the original transport equation. Then to retrieve $|W(t)\rangle$, it can be shown that the query and gate complexities are respectively $\sim \tilde{\mathcal{O}}((\|\bb{W}(0)\|/\|\bb{W}(t)\|)s(\bb{\sigma})/\epsilon)$ and $\sim \tilde{\mathcal{O}}((\|\bb{W}(0)\|/\|\bb{W}(t)\|)ds(\bb{\sigma})/\epsilon)$, where $s(\bb{\sigma})$ is the sparsity of the differential cross-section matrix. This approach can also be applied to find the stationary state for $|W(t)\rangle$. 

Here $d \geq 3$ and refers to the  spatial dimension for Schr\"odinger's equation. It is also possible to interpret this in terms of multiple $n=d/3$ particles each moving in three spatial dimensions so a high $d$ limit is a large $n$ limit. See Appendix D. 

While $|W(t)\rangle$ not the usual output of a quantum simulator, which outputs directly a wavefunction $|\psi(t)\rangle$, it is still possible to recover the semiclassical limit of the physical observables of $|\psi(t)\rangle$ by taking the moments of $W$. Namely, any expectation value $\langle \hat{G}\rangle=\langle \psi(t)|\hat{G}|\psi(t)\rangle=\iint dx dk W(t,x,k) g(x, k)$, where  $g(x,k)=1, k, k^2/2$, gives mass, momentum and energy respectively.  Given access to the quantum state $|g(x,k)\rangle$ whose amplitudes are proportional to $g(x,k)$ at mesh points $(x_i, k_j)$, $\langle \hat{G}\rangle$ is then proportional to the quantum fidelity between $|g(x,k)\rangle$ and $|W(t,x,k)\rangle$. This fidelity can be recovered in different ways, for instance, via a quantum swap-test \cite{buhrman2001quantum}.

\section{General formulation}

The Schr\"odingerisation approach can in fact be generalised to \textit{any} linear $(d+1)$-dimensional PDE  for $u(t,x)$. We study the PDE that is first-order in time. A system with higher-order derivatives in $t$ can be written as an enlarged system with first-order time-derivatives by introducing a new variable $u'=\partial_t u$. Using the same discretisation scheme as before, where the entries of vector $\bb{u}(t)$ with size $M^d$ are solutions to a homogeneous linear PDE on the grid points $(t, x_i)$, we are left with a system of linear ODEs
\begin{align}
d\bb{u}(t)/dt=-\bb{A} \bb{u} (t),
\end{align} 
where the matrix $\bb{A}$ is in general not Hermitian, unlike in the heat equation example. However, it is always possible to make the decomposition $\bb{A}=\bb{H}+i\bb{\bar{H}}$ in terms of Hermitian matrices $\bb{H}=(\bb{A}+\bb{A}^{\dagger})/2$ and $\bb{\bar{H}}=i(\bb{A}^{\dagger}-\bb{A})/2$. To ensure stability, we can assume $\bb{H}$ to be positive semi-definite. The heat equation example belongs to the case $\bb{\bar{H}}=0$, whereas purely Schr\"odinger dynamics corresponds to $\bb{H}=0$. A generalisation to inhomogeneous linear PDEs is possible by a straightforward dilation of $\bb{u}$ and $\bb{A}$. This is also useful for solving boundary value problems \cite{schr2} by quantum simulation. 

We can apply our warped phase transformation here to $\bb{u}(t)$ in the $p>0$ region, $\bb{v}(t,p)=e^{-p}\bb{u}(t)$, where the original solution can be recovered by using $\bb{u}(t)=\int_{-\infty}^{\infty} \chi(p)\bb{v}(t,p)dp$. Just like for the heat equation, it is possible to extend the initial condition to $p<0$ by defining $\bb{v}(0,p)=\exp(-|p|)\bb{u}(0)$ and in the $p \in (-\infty, \infty)$ region $\bb{v}(t,p)$ satisfies $\partial_t \bb{v}+\bb{H}\partial_p \bb{v}-i\bb{\bar{H}}\bb{v}=0$. Let $\tilde{\bb{v}}(t, \eta)$, $\eta \in \mathbb{R}$ be the Fourier transform of $\bb{v}$ in $p$, we obtain the generalisation of Eq.~\eqref{eq:schrodinger}
\begin{align}
i\partial_t \tilde{\bb{v}}=(\eta \bb{H}+\bb{\bar{H}})\tilde{\bb{v}}
\end{align}
which is still a system of Schr\"odinger equations, one for each $\eta$, since $\eta \bb{H}+\bb{\bar{H}}$ is Hermitian. Note this form is not, in a strict sense,  Schr\"odinger's equations per se, where the Hamiltonian is usually a Laplacian plus a potential function. However, the first order differential operator can be cast explicitly into the Schr\"odinger form by making another warped phase transformation. We omit the details here since it's not crucial for quantum simulation to be possible. As previously, we can proceed by discretising $\eta$ to obtain 
\begin{align}
i\frac{d}{dt}\tilde{\bb{v}}=(\bb{H}\otimes \bb{D}+\bb{\bar{H}}\otimes \mathbf{1})\tilde{\bb{v}}=\bb{H}_{\text{total}}\tilde{\bb{v}},
\end{align}
and the state evolves $\tilde{\bb{v}}(t)=\exp(i\bb{H}_{\text{total}}t)\tilde{\bb{v}}(0)$ with respect to unitary dynamics. Let $s \sim \max(s(\bb{H}), s({\bb{\bar{H}}}))$ be the maximum of the sparsity of $\bb{H}$ and $\bb{\bar{H}}$. The max-norm $\|\bb{H}_{\text{total}}\|_{\text{max}} \sim \max(\|\bb{H}\|_{\text{max}}/\epsilon, \|\bb{\bar{H}}\|_{\text{max}})$. For instance, in the transport equation case, the two terms in the latter expression are of the same order. Then we can simulate $|u(t)\rangle$ with the following resources, where the proof follows in the same way as Theorem~\ref{thm:heatequation} with $D=M^d$.
\begin{theorem} \label{thm:general}
Given sparse access to the $M^d \times M^d$ matrix $\bb{H}_{\text{total}}$ and the unitary $U_{initial}$ that prepares the initial quantum state $|\bb{u}(0)\rangle$ to precision $\epsilon$. With the Schr\"odingerisation approach, the state $|\bb{u}(t)\rangle$ can be prepared with query complexity $\tilde{\mathcal{O}}((\|\bb{u}(0)\|/\|\bb{u}(t)\|)st\|\bb{H}_{\text{total}}\|_{\text{max}})$ and $\tilde{\mathcal{O}}(\|\bb{u}(0)\|/\|\bb{u}(t)\|dts\|\bb{H}_{\text{total}}\|_{\text{max}})$ additional two-qubit gates.
\end{theorem}

\section{Discussion}

We have introduced a conceptually new method, called Schr\"odingerisation, which makes it possible to simulate solutions of any linear partial differential equations using quantum simulation. This method is framed in the traditional language of dynamical equations. The connection between classical dynamics and its corresponding Schr\"odinger's equations can be seen more easily at the level of the dynamical equations without more sophisticated methods. 

For more in-depth technical details, discussions and applications to other linear PDEs like Liouville, Fokker-Planck, Vlaslov-Fokker-Planck, Black-Scholes equations and nonlinear ODEs, please see our technical companion paper \cite{schr2}. 
\section{Acknowledgements}
SJ was partially supported by the NSFC grant No.~12031013, the Shanghai Municipal Science and Technology Major Project (2021SHZDZX0102), and the Innovation Program of Shanghai Municipal Education Commission (No. 2021-01-07-00-02-E00087).  NL acknowledges funding from the Science and Technology Program of Shanghai, China (21JC1402900). 
\bibliography{QHS}

\appendix 

\section{Appendix A: Warped phase transformation} \label{app: Appendix A}

Here we define in the $p>0$ region 
\[w(t,x,p) = \e^{-p} u(t,x), \qquad p>0.\]
A simple calculation shows
\begin{equation}\label{heatreformulation}
\partial_t w + \partial_p \nabla_x w = 0,  \qquad p>0.\\
\end{equation}
We can recover $u$ using $w$ via
\begin{equation}\label{integration}
  u(t,x)= \int_0^\infty w(t,x,p)\d p = \int_{-\infty}^\infty \chi(p) w(t,x,p) \d p,
  \end{equation}
where $\chi(p) = 1$ for $p>0$ and $\chi(p) = 0$ for $p<0$. Alternatively, one can also choose any $p_*>0$ and let
\begin{equation} \label{point}
  u(t,x) = \e^{p_*} w(t,x, p_*).
\end{equation}
Now we want to see that no boundary condition is required at $p=0$. Assume $x$ is defined in a periodic domain and we apply the Fourier transform on $w$ with respect to $x$, denoted $\hat{w}(t,\xi,p)$, with Fourier modes $\xi=(\xi_1,...,\xi_d)^T$. One then arrives at a convection equation $\partial_t \hat{w}- |\xi|^2 \partial_p \hat{w} = 0$, where $|\xi|^2 = \xi_1^2 + \cdots + \xi_d^2$. Here the solution $\hat{w}$ moves from the right to the left, so no boundary condition is needed at $p = 0$.

For the convenience of numerical approximation,  we extend the domain to $p<0$. This extension will not affect the solution of  $w$ in the region $p>0$ since $w$ convects from right to the left.  Thus, we can symmetrically extend the initial data of $w$ to $p<0$ and keep Eq.~\eqref{heatreformulation}:
\begin{equation}\label{heatreformulationextend}
\begin{cases}
\partial_t w +  \Delta_x \partial_p w= 0, \qquad p \in (-\infty, \infty), \\
w(0,x,p) = \e^{-|p|} u_0(x).
\end{cases}
\end{equation}
The only difference is that now the initial condition is $w(0,x,p)=\exp(-|p|)u_0(x)$ for $p \in (-\infty, \infty)$ instead of just $w(0,x,p)=\exp(-p)u_0(x)$ for $p>0$. 

The solution coincides with the solution of Eq.~\eqref{heatreformulation} when $p>0$. Due to the exponential decay of $\e^{-|p|}$ one can (computationally) impose the periodic boundary condition $w(t,x,p=-L) = w(t,x,p = L)= 0$ along the $p$-direction for some $L>0$ suitably large but finite. Then the Fourier transform on $w$ with respect to
$p$ gives
\begin{equation}\label{heat-Schro}
  \partial_t\tilde{w} -\i \eta \nabla^2_x\tilde{w}=0 \qquad \mbox{or} \qquad \i \partial_t\tilde{w} = - \eta \nabla^2_x\tilde{w},
\end{equation}
where $\tilde w(t,x,\eta)$, $\eta\in \mathbb{R}$, is the Fourier transform of $w$ in $p$.  Equation \eqref{heat-Schro} is clearly the Schr\"odinger equation, for every $\eta$!

For more details and numerical experiments on justifying the setup in Eq.~\eqref{heatreformulationextend}, see our technical companion paper \cite{schr2}.

\section{Appendix B: Proof of Theorem~\ref{thm:heatequation}} \label{app: Appendix B}

From Eq.~\eqref{eq:schrodinger} one can prepare the $\tilde{\mathcal{O}}(\log(D))$-qubit state $|\tilde{w}(t)\rangle=\exp(-i\bb{H}_{\text{total}}t)|\tilde{w}(0)\rangle$ by simulating the Hamiltonian $\bb{H}_{\text{total}}=\bb{H}\otimes \bb{D}$ to time $t$. Here the results hold so long as $\bb{H}$ is Hermitian. Let $s$ and $\|\bb{H}\|_{\text{max}}$ be the sparsity and max-norm of a more general $D \times D$ Hermitian matrix $\bb{H}$. 

Then, given $|\tilde{w}(0)\rangle$, one can apply Lemma~\ref{eq:lemmasimulation}, so that one prepares $|\tilde{w}(t)\rangle$ with query and gate complexity $\tilde{\mathcal{O}}(s\|\bb{H}_{\text{total}}\|_{\text{max}}t)$, where $\tilde{\mathcal{O}}$ denotes that logarithmic terms, like $\log (D)$ and $\log(1/\epsilon)$, are ignored. Now $\|\bb{H}_{\text{total}}\|_{\text{max}}=\|\bb{H}\|_{\text{max}}\|\bb{D}\|_{\text{max}}$, where $\|\bb{D}\|_{\text{max}}=\mathcal{O}(N)$ by definition and $N=\mathcal{O}(1/\epsilon)$. Thus $\tilde{\mathcal{O}}(st\|\bb{H}_{\text{total}}\|_{\text{max}})=\tilde{\mathcal{O}}(st\|\bb{H}\|_{\text{max}}/\epsilon)$.

From $|\tilde{w}(t)\rangle$ that results from the Hamiltonian simulation problem, 
the simplest way to prepare $|u(t)\rangle$ from $|\tilde{w}(t)\rangle$ is to project $|\tilde{w}(t)\rangle$ with operator $\hat{P}=\mathbf{1}\otimes \sum_{k=N/2}^N |k\rangle \langle k|$ (which projects only onto $p>0$ states). This leads to 
\begin{align}
\hat{P} |w(t)\rangle=\frac{\|u(t)\|\|\exp(-p) \|}{\|w(t)\|}|\exp(-p)\rangle|u(t)\rangle
\end{align}
where $|\exp(-p)\rangle=(1/\|\exp(-p)\|)\sum_{k=N/2}^N \exp(-p_k)|k\rangle$ with normalisation constant $\|\exp(-p)\|^2=\sum_{k=N/2}^N \exp(-2p_k) \sim \mathcal{O}(N \int_0^{\infty}\exp(-2p) dp)=\mathcal{O}(N)$. Thus a simple projection retrieves $|u(t)\rangle$ with probability $(\|u(t)\|\|\exp(-p)\|/\|w(t)\|)^2\sim N (\|u(t)\|/\|w(t)\|)^2$. We also used $\|\tilde{w}(t)\|=\|w(t)\|$. 

One can also use amplitude amplification to boost this probability to $\sqrt{N}\|u(t)\|/\|w(t)\|$ in the usual manner, by assuming access to the oracle $Q=-S_wS_p$ where $S_w=\mathbf{1}-|w(t)\rangle \langle w(t)|$ and $S_p=\mathbf{1}-2\hat{P}$ \cite{brassard2002quantum}. This means $\tilde{\mathcal{O}}(\|w(t)\|/(\sqrt{N}\|u(t)\|))$ queries to $Q$ is sufficient. Combining this with Lemma~\ref{eq:lemmasimulation}, we see that this requires a total query complexity of $\tilde{\mathcal{O}}(st\|\bb{H}\|_{\text{max}} \|w(t)\|/(\sqrt{N}\|u(t)\|))$. 

It can be shown, using the quadrature rule, that $(1/(M^d N))\|w(t)\|^2 \approx \iint_{-\infty}^{\infty} dx dp |w(t,x,p)|^2=\int_{-\infty}^{\infty}|u(0)|^2 \approx (1/M^d)\|u(0)\|^2$. Thus $\|w(t)\|/(\sqrt{N}\|u(t)\|)=\|u(0)\|/\|u(t)\|$.

In the above, we assumed access to $|\tilde{w}(0)\rangle$, while it is more natural to assume access to $|u(0)\rangle$. It is straightforward to obtain the initial state $|\tilde{w}(0)\rangle$ from $|u(0)\rangle$. Since $w(0,x,p)=\exp(-|p|)u(0,x)$, then $|w(0)\rangle=|u(0)\rangle |\exp(-p)\rangle$, where we assume we can prepare the state $|\exp(-p)\rangle$. Then applying the quantum Fourier transform to $|w(0)\rangle$, with respect to $p$, gives $|\tilde{w}(0)\rangle$. To perform an inverse Fourier transform, only $\tilde{\mathcal{O}}(1)$ two-qubit gates are required, since we have only $\log(N)$ qubits in the $p$-variable. The number of qubits in $|\tilde{w}(t)\rangle$ is $\log(DN)$. 

Thus, when given $|u(0)\rangle$, the total query and gate complexity to prepare $|u(t)\rangle$ is $\tilde{\mathcal{O}}((\|u(0)\|/\|u(t)\|)st\|\bb{H}\|_{\text{max}}/\epsilon)$.

\section{Appendix C: Preparation of quantum ground state via  Schr\"odingerisation} \label{app: Appendix C}
Given a (Hermitian) Hamiltonian $\bb{H}$, an important question is how to prepare its ground state. If given a quantum state $|u(0)\rangle$, where $\|u(0)\|=1$, then unitary evolution with respect to a time-independent $\bb{H}$ for time $\tau$ gives $|u(\tau)\rangle=\exp(-i\bb{H}\tau)|u(0)\rangle$. Here it is not clear at what time $t$ one might be able to retrieve the ground state, if at all. So an extension to non-unitary evolution would be helpful.

One method is the imaginary time evolution approach, which maps $\tau \rightarrow t=i \tau$. In this case, one can write the non-unitary evolution as 
\begin{align} \label{eq:heatsolutionnormalised}
    |u(t)\rangle=\frac{e^{-\bb{H}t}}{\|u(t)\|}|u(0)\rangle.
\end{align}
Since this is now non-unitary evolution $\|u(t)\| \neq \|u(t=0)\|$. However, there is now no direct way to transform the solution $|u(\tau)\rangle$ into Eq.~\eqref{eq:heatsolutionnormalised} since there is no $\tau \rightarrow t$ conversion that can be used once the solutions are obtained. The two solutions will only agree in the stationary limit $\partial u/\partial t=0$. One can instead use an indirect method \cite{motta2020determining}, where for each $t$, one finds the corresponding Hamiltonian evolving with $\tau$, by making tomographic measurements of $|u\rangle$ at every small time-step. We seek a more direct method, without intervening the quantum system during the process, that still allows us to perform this non-unitary evolution. 

What we lose in unitarity of evolution in Eq.~\eqref{eq:heatsolutionnormalised}, we gain in the simplicity of seeing ground states emerge. Let the orthonormal eigenfunctions of $\bb{H}$ be $\{E_j\rangle\}$ for $j=1,...,D$. We can thus decompose $|u(0)\rangle$ 
\begin{align}
    |u(0)\rangle=\sum_{j=0}^{D-1} \alpha_j |E_j\rangle, \quad \alpha_j \in \mathbb{C}
\end{align}
with corresponding eigenvalues $E_j$ where $\sum_{j=0}^{D-1} |\alpha_j|^2=1$. Then it is clear, if we evolve non-unitarily, 
\begin{align}
    |u(t)\rangle=\frac{e^{-\bb{H}t}}{\|u(t)\|}|u(0)\rangle=\frac{1}{\|u(t)\|)}\sum_{j=0}^{D-1} \alpha_j e^{-E_j t}|E_j\rangle
\end{align}
This means that for low eigenenergies, the decay is the slowest and $|E_0\rangle$ would emerge as $t\rightarrow \infty$ for gapped Hamiltonians. Furthermore, the convergence rate toward the ground state is exponential with the decay rate of $\exp((E_0-E_1)t)$, if there is a spectral gap, i.e., $E_1>E_0$.

Here we can use Schr\"odingerisation to obtain the state $|u(t)\rangle \propto \exp(-\bb{H}t)|u(0)\rangle$. To specify how long $t$ we need to evolve, let $t_{\text{final}}$ be the minimum evolution time so that quantum fidelity between the prepared $|u(t_{\text{final}})\rangle$ state and the true ground state $|E_0\rangle$ is greater or equal to $1-\epsilon$, i.e., 
\begin{align} \label{eq:fidelitycondition}
|\langle u(t_{\text{final}})|E_0\rangle|^2=\frac{|\alpha_0|^2e^{-2E_0 t_{\text{final}}}}{\|u(t_{\text{final}})\|^2}\geq 1-\epsilon
\end{align}
for some $\epsilon \ll 1$. For a normalised initial state $\|u(0)\|=1$ we have
\begin{align} \label{eq:psitaunorm}
    \|u(t)\|=\|e^{-\bb{H}t}|u(0)\rangle\|.
\end{align}
thus $\|u(t)\|^2=|\alpha_0|^2\exp(-2E_0 t)(1+(|\alpha_1|/|\alpha_0|)^2 \exp(-2\Delta t))+L$ where $0<\Delta=E_1-E_0$ is the spectral gap and we can assume $L=\sum_{k \ge 2} |\alpha_k|^2e^{-2tE_k}\ll 1$ for $t \gg 1$. Then for $\delta \ll 1$, one can insert this into Eq.~\eqref{eq:fidelitycondition} and obtain
\begin{align}
t_{\text{final}} \gtrsim (1/2\Delta)\ln(|\alpha_1|^2/(\epsilon |\alpha_0|^2) \sim (1/\Delta)\ln(1/(\epsilon |\alpha_0|^2))
\end{align}
when $|\alpha_k| \ll 1$ for all $k\geq 2$. 

Inserting $t=t_{\text{final}}$ in Theorem~\ref{thm:heatequation} and $E_0 \ll t_{\text{final}}$, one gets the total query and gate complexity $\tilde{\mathcal{O}}(s\|\bb{H}\|_{\text{max}}/(|\alpha_0|\Delta \epsilon))$ for preparing the ground state of $\bb{H}$. 

The proof for Gibbs state preparation proceeds similarly, except there $t \propto 1/T$ where $T$ is the temperature of the Gibbs state. 

\section{Appendix D: Transport equation and quantum systems in random media} \label{app: Appendix D}

Energy conservation for general wave problems is given by the transport equation for the scalar energy density 
\begin{align} \label{eq:transport}
\partial_t f + \nabla_{k}\omega \cdot \nabla_x f-\nabla_x\omega \cdot \nabla_{k}f=S,
\end{align}
with initial condition $f(0,x,k) = f_0(x,k)$, 
where $f = f(t,x,k)$ is the scalar energy density and $\omega=\omega(x,k)$ is the frequency of the wave with wave vector $k \in \mathbb{S}^{d-1}$ at point $x \in [-1,1]^d$. We assume periodic boundary conditions are imposed. Here $S=\int dk' (\sigma(x,k, k')f(t,x,k')-\Sigma(x,k)f)$ where $\sigma=\sigma(x,k,k')$ is the differential scattering cross section and $\Sigma(x, k)=\int dk'\sigma(x,k',k)$ is the total cross-section. 

An interesting example of this equation for a quantum problem appears in the simulation of quantum states in random media in the semi-classical regime. Here we study the evolution of a quantum system in the presence of a random potential $U(x)=U_0(x)+U_1(x/\zeta)$ where $U_0(x)$ is a deterministic  slowly moving background potential, $U_1$ is the random fluctuation depending on  $x/\zeta$, the fast spatial variable.

Let $R(|x-y|)=\langle U_1(x)U_2(y)\rangle$ be the covariance of the fluctuations, which are assumed to be spatially homogeneous and isotropic, and $R(|k|)=(1/(2\pi))^d\int dk \exp(iky) R(x)$ is the power spectrum. The differential cross-section and the total cross-section in this case are respectively $\sigma(k,k')/(4\pi)=R(k-k')\delta(k^2-(k')^2)$ and $\Sigma(k)/(4\pi)=\int dk' R(k-k')\delta(k^2-(k')^2)$. For the Schr\"odinger equation, $\omega^2=k^2/2+U_0(x)=k^2/2$, we can ignore the slowly-moving background potential $U_0(x)$.

 The semiclassical or weakly coupling limit of the  Wigner function  for the quantum system in this random medium, which is the particle density distribution,  solves the transport equation in Eq.~\eqref{eq:transport} \cite{ryzhik1996transport} and for any $d \geq 3$ the above forms of the cross-sections hold (Born approximation) \cite{erdHos2000linear} and
 \begin{align} \label{eq:wignerpde} 
\partial_t W+k \cdot \nabla_x W=\int dk'\sigma(k,k')W(t,x,k')-\Sigma(k) W.
 \end{align}
We can apply the warped phase transformation $V=V(t,x,k,p)=\exp(-p)W$ to get
\begin{align} \label{eq:vequation}
\partial_t V+k\cdot \nabla_x V=-\int dk'\sigma(k,k')\partial_pV(t,x,k',p)+\Sigma(k) \partial_p V.
\end{align}
Now apply a Fourier transform with respect to \textit{both} the $x$ and $p$ variables, which have respective conjugate variables $\xi \in \mathbb{R}^d$ and $\eta \in \mathbb{R}$ and define this Fourier transform as
\begin{align}
\tilde{V}=\tilde{V}(t,\xi, k, \eta)=\int dx dp\, e^{-i\eta p}e^{-i\xi x}V(t,x,k,p).
\end{align}
Using Eq.~\eqref{eq:vequation}, it is simple to see that $\tilde{V}$ obeys
\begin{align} \label{eq:tildev}
i\partial_t \tilde{V}=(k \cdot \xi-\eta \Sigma(k))\tilde{V}+\eta \int dk'\sigma(k,k')\tilde{V}(t,\xi,k',\eta).
\end{align}
We discretise in $\xi, k$ and $\eta$ with $\Delta \xi=2/J$, $\Delta k=2/K$ and $\Delta \eta=2/N$ for positive even integers $J, K, N$. Let $\bb{\tilde{V}}(t)=\bb{\tilde{V}}_{[\xi]}\otimes \bb{\tilde{V}}_{[k]}\otimes \bb{\tilde{V}}_{[p]}$ be the vector formed by the discretisation of $\tilde{V}$ in $\xi, k$ and $\eta$. Here each element of $\bb{\tilde{V}}_{[x]}$ denotes $\tilde{V}(t,\xi_i, k,\eta)$ with $i=1,...,J^d$ labelling the $\xi$ grid points, each element of $\bb{\tilde{V}}_{[k]}$ is $\tilde{V}(t,\xi,k_j,\eta)$ with $j=1,...,K^d$ labelling the $k$ grid points and the elements of $\bb{\tilde{V}}_{[\eta]}$ are $\tilde{V}(t,\xi,k,\eta_k)$ with $k=1,...,N$. One can then use the quadrature rule for the last term of Eq.~\eqref{eq:tildev} to write
\begin{align}
\int dk'\sigma(k,k')\tilde{V}(t,\xi,k',\eta) \approx \sum_{j'=1}^{K^d} \sigma_{j, j'} \tilde{V}(t,\xi, k'_{j'}, \eta)
\end{align}
where $\sigma_{j,j'}=\sigma(k_j, k'_{j'})/K^d$ are the elements of a symmetric (and Hermitian) matrix $\bb{\sigma}$, since we are looking at isotropic scattering. 
Then from Eq.~\eqref{eq:tildev} the vector $\bb{\tilde{V}}(t)$ obeys 
\begin{align} \label{eq:boltzmannfinal}
i\frac{d \bb{\tilde{V}}}{dt}=\bb{H}_{\text{total}}\bb{\tilde{V}}
\end{align}
where the matrix $\bb{H}$ is Hermitian 
\begin{align} \label{eq:hboltzmann}
\bb{H}_{\text{total}}=\bb{L}_{[\xi, k]} \otimes \bb{1}-\mathbf{1}\otimes \bb{\Sigma} \otimes \bb{D}+\mathbf{1}\otimes \bb{\sigma} \otimes \bb{D}.
\end{align}
Here $\bb{D}=\text{diag}(\mu_1,...,\mu_N)$ with $\mu_k=\pi(k-N/2)$ and $\bb{\Sigma}=\text{diag}(\Sigma(k_1),...,\Sigma(k_{K^d}))$. The matrix $\bb{L}_{[\xi, k]}$ acts on $\bb{\tilde{V}}_{[\xi]}\otimes \bb{\tilde{V}}_{[k]}$ and multiplies its $(i,j)^{\text{th}}$ component it by factor $\xi_i \cdot k_j$.

Since $\bb{H}_{\text{max}}$ is Hermitian, one can use Hamiltonian simulation methods in Eq.~\eqref{eq:boltzmannfinal} where $\bb{\tilde{V}}(t)=\exp(-i\bb{H}_{\text{total}}t)\bb{\tilde{V}}(0)$ and one can prepare the quantum state $|\tilde{V}(t)\rangle=(1/\|\bb{\tilde{V}}(t)\|)\sum_{i=1}^{J^d}\sum_{j=1}^{K^d}\sum_{k=1}^N \bb{\tilde{V}}(t)|i,j,k\rangle$ by just evolving the Hamiltonian $\bb{H}_{\text{total}}$, when given the initial state $|V(0)\rangle$. Note that $\|\bb{\tilde{V}}(t)\|=\|\bb{\tilde{V}}(0)\|$ since $\bb{\tilde{V}}(t)$ evolves by a unitary transformation. This can be used to prepare the quantum state whose amplitude is proportional to the average Wigner function of the quantum state in a random medium in the following way. 

Let $\bb{W}(t)$ be the vector whose elements are  $W(t,x_i, k_j)$ at grid points $(x_i, k_j)$. Suppose we want to prepare the quantum state $|W(t)\rangle=(1/\|\bb{W}(t)\|)\sum_{i=1}^{M^d}\sum_{j=1}^{K^d}\bb{W}(t)|i,j\rangle$. By definition, $\bb{V}(t)=(\mathcal{F}^{-1}_{x}\otimes \mathbf{1}\otimes \mathcal{F}^{-1}_p)\bb{\tilde{V}}(t)$, where $\mathcal{F}^{-1}_{(\cdot)}$ is the inverse Fourier transform with respect to variable $(\cdot)$ and $\bb{V}(t)$ is the discretisation of the warped phase transformation $V(t,x,k,p)$. Then $|V(t)\rangle=(1/\|\bb{V}(t)\|)\sum_{i=1}^{M^d}\sum_{j=1}^{K^d}\sum_{k=1}^N \bb{W}(t) \exp(-p_k)|i,j,k\rangle$. Fourier transforms do not change norms, so $\|\bb{V}(t)\|=\|\bb{\tilde{V}}(t)\|=\|\bb{\tilde{V}}(0)\|=\|\bb{V}(0)\|$. 
We can use the projection $\hat{P}$ in Appendix B and show $|W(t)\rangle$ and $|\tilde{V}(t)\rangle$ are related by
\begin{align}
& \hat{P}(\mathcal{F}^{-1}_x\otimes \mathbf{1}\otimes \mathcal{F}^{-1}_p)|\tilde{V}(t)\rangle=\frac{\|\bb{W}(t)\|}{\|\bb{V}(0)\|}|W(t)\rangle\sum_{k=1}^N e^{-p_k}|k\rangle \\ \nonumber 
&=\frac{\|\bb{W}(t)\|\|\exp(-p)\|}{\|\bb{V}(0)\|}|W(t)\rangle |\exp(-p)\rangle 
\end{align}
where $|\exp(-p)\rangle$ is the state with normalisation constant $ \| \exp(-p)\| \sim \mathcal{O}(\sqrt{N})$, as defined in Appendix B. Then following similar methods in Appendix B, it can be proved that $\|\bb{V}(0)\|/\|\bb{W}(t)\| \sim \sqrt{N}\|\bb{W}(0)\|/\|\bb{W}(t)\|$.



The proof of the resource cost in obtaining $|W(t)\rangle$ follows in the same way as shown in Appendix B. Let $s$ be the sparsity of $\bb{H}_{\text{max}}$ and $\|\bb{H}_{\text{total}}\|_{\text{max}}$ is its max-norm. Then following Lemma~\ref{eq:lemmasimulation}, we see that, given $|W(0)\rangle$, the query and gate complexities to obtain $|W(t)\rangle$ are respectively $\tilde{\mathcal{O}}((\|\bb{W}(0)\|/\|\bb{W}(t)\|)st \|\bb{H}_{\text{total}}\|_{\text{max}})$ and $\tilde{\mathcal{O}}((\|\bb{W}(0)\|/\|\bb{W}(t)\|)dst \|\bb{H}_{\text{total}}\|_{\text{max}})$.

From Eq.~\eqref{eq:hboltzmann}, one sees that the sparsity $s$ is dominated by $s(\bb{\sigma})$, which is the sparsity of $\bb{\sigma}$. This means the more off-diagonal terms in $\bb{\sigma}$ (corresponding to more scattering to different wave vectors), the larger the cost in preparing $|W(t)\rangle$. To find the max-norm $\|\bb{H}_{\text{total}}\|_{\text{max}}$, we note that $\|\bb{L}_{[\xi, k]}\|_{\text{max}}\sim \|\bb{\sigma}\otimes \bb{D}\|_{\text{max}} \sim \|\bb{\sigma}\otimes \bb{D}\|_{\text{max}} \sim \mathcal{O}(1/\epsilon)$ where $\bb{\sigma}$ and $\bb{\Sigma}$ are order one factors. This means that  the Schr\"odingerisation approach {\it does not change the order of $\|\bb{H}_{\text{total}}\|_{\text{max}}$}, since the transformed equation remains first order as the original equation.

We can also apply this to finding the stationary state of the semiclassical limit of the Wigner function, i.e., $|W(t_{s})\rangle$ when $d\bb{W}(t)/dt\vert_{t=t_s}=0$. The analysis proceeds in the same way to preparing the ground state in Appendix C. The transport operator in Eq.~\eqref{eq:transport} is dissipative, or rather hypocoercivity, with spectral gap shown in \cite{Mouhot}. In the case of isotropic scattering, the spectra includes $0$ and all other strictly negative ones, depending on Chandrasekhar's H-function \cite{Chand}. We can thus use $E_0=0$ in this case and the same cost as in Appendix C follows.

We can also extend the above analysis in the high $d$ limit by reinterpreting as a multiparticle $n=d/3$ limit with each particle moving in three spatial dimensions. However, in this case, $\omega^2=k^2/2+U_C$ since the interparticle interactions with a potential $U_C$ (e.g. Coulomb) needs to be included. This means we need to augment Eq.~\eqref{eq:wignerpde} by an extra term
 \begin{align} \label{eq:wignerpde2} 
& \partial_t W+k \cdot \nabla_x W-\nabla_x U_C \cdot \nabla_k W \nonumber \\
&=\int dk'\sigma(k,k')W(t,x,k')-\Sigma(k) W.
 \end{align}
Here we omit the precise details (see \cite{schr2}), but one can discretise in $x$, $k$ and rewrite this as a system of linear ODEs and apply the general formalism.

\end{document}